\documentclass[11pt]{article}

\pdftrailerid{}
\usepackage[T1]{fontenc}
\usepackage[utf8]{inputenc}
\usepackage{lmodern,amsmath,amssymb,amsthm,mathtools,mathrsfs}
\usepackage[a4paper,margin=1.02in]{geometry}
\usepackage{booktabs,longtable,array,microtype}
\usepackage{xcolor}
\usepackage[colorlinks=true,linkcolor=blue!50!black,
  citecolor=blue!50!black,urlcolor=blue!50!black]{hyperref}
\usepackage{enumitem,needspace}
\setlist[itemize]{leftmargin=1.5em,itemsep=2pt,topsep=4pt}
\setlist[enumerate]{leftmargin=1.8em,itemsep=2pt,topsep=4pt}
\allowdisplaybreaks
\numberwithin{equation}{section}
\newtheorem{theorem}{Theorem}[section]
\newtheorem{proposition}[theorem]{Proposition}

\newtheorem{corollary}[theorem]{Corollary}
\theoremstyle{definition}

\theoremstyle{remark}

\newcommand{\R}{\mathbb R}

\newcommand{\Q}{\mathbb Q}

\newcommand{\dd}{\,\mathrm d}

\newcommand{\sgn}{\operatorname{sgn}}

\hypersetup{
  pdftitle={Closed-form evaluation of the cyclic four-electron atomic correlation integral},
  pdfauthor={Jonas Matuzas},
  pdfsubject={Classical-polylogarithmic evaluation and certified verification of a cyclic four-electron Slater integral},
  pdfkeywords={four-electron atomic integral, Hylleraas basis, polylogarithm,
    Feynman period, rigorous numerical verification}}
\title{Closed-form evaluation of the cyclic\\four-electron atomic correlation integral}
\author{Jonas Matuzas\\
\href{mailto:jonas.matuzas@gmail.com}{\nolinkurl{jonas.matuzas@gmail.com}}}
\date{}
\begin{document}
\maketitle
\begin{abstract}
For positive exponents $\alpha_1,\ldots,\alpha_4$, we evaluate the cyclic four-electron Slater integral
\[
 I(\alpha)=\int_{\R^{12}}\frac{e^{-\sum_i\alpha_i r_i}}
 {r_1r_2r_3r_4\,r_{12}r_{23}r_{34}r_{41}}\prod_{i=1}^4\dd^3r_i
\]
in a finite real expression of weight at most three.  The general answer is an explicit sum of eighteen blocks involving only logarithms, dilogarithms, and trilogarithms of rational functions of the exponents; all logarithm arguments are positive and all nonconstant polylogarithm arguments may be placed in $(-1,1)$.  At equal exponents,
\[
 I(\lambda,\lambda,\lambda,\lambda)=\frac{(4\pi)^4}{\lambda^4}
 \left\{\frac12\operatorname{Li}_3\!\left(\frac14\right)-\frac7{24}\zeta(3)
 +\frac{\pi^2\log2}{12}-\frac23\log^3 2\right\}.
\]
The equal-exponent value is a weight-three period over $\Q$ ramified only at $2$ and $3$.  Three independently implemented numerical routes agree at four mass tuples to at least 32 decimal places, and an exact-dyadic evaluation of the classical formula gives eight 100-place certificates.  We also prove the small- and large-exponent limits, dihedral symmetry, and a strict rearrangement law.  The formula applies to the four-cycle and does not evaluate the two chorded topologies.
\end{abstract}

\noindent\textbf{Generative-AI disclosure and author responsibility.}
OpenAI GPT-5.6 Pro was used for derivations, computations, and drafting; Claude Fable 5.1 was used for independent re-derivation and audit of every increment. The author directed and audited the work throughout and takes full responsibility for the mathematics and the final text.



\section{Introduction}
% TRACE: history and scope -- inspected literature; C4EXEC publication audit.
The integral studied here is the first four-electron rim topology whose cycle passes through all four electrons.  In Harris's terminology, it is the smallest linked topology beyond three vertices.  Its rim is the cycle $C_4=\{12,23,34,41\}$, while the nucleus is joined to each electron by a Slater spoke.  The result below concerns this topology only.  Adding either diagonal changes the analytic class of the fixed-radius angular problem.

Bonham's 1965 classification isolates, as case V, the four-electron cyclic integral and formulates conjectures A--C for reductions of the corresponding angular and radial expressions \cite{Bonham1965}.  Roberts treated related four-electron correlated integrals in the same period \cite{Roberts1965}.  Bonham's 1966 note treats the square with one diagonal---the chorded topology of Section~\ref{sec:outlook}---and leaves it in the form of a multiple integral \cite{Bonham1966}.

For three electrons, Fromm and Hill obtained the general master integral in dilogarithms \cite{FrommHill}; Remiddi gave a compact analytic form \cite{Remiddi}, Harris supplied another formulation \cite{Harris1997}, and Pachucki, Puchalski, and Remiddi derived recursions for the nonnegative-index family \cite{PPR2004}.  These results solve the triangle rim but do not supply the four-cycle.

King developed systematic four-electron reductions and closed examples in 2014 and 2016 \cite{King2014,King2016}.  His tables include forests and a triangle with a leaf, but not the cyclic integral evaluated here.  Harris's review explains that standard Hylleraas--configuration interaction attaches at most one interelectronic factor to each configuration; its integral algebra therefore does not generate a ring \cite{Harris2018}.  The present result is instead relevant to fully correlated Hylleraas expansions, double-linked Hy-CI, and extended Hy-CI.

The momentum-space counterpart is a three-dimensional vacuum integral.  Broadhurst evaluated the massive three-loop Ising tetrahedron in Clausen/dilogarithmic form and remarked that the result suggested a polylogarithmic ladder beginning with $\pi/\sqrt2$ \cite{Broadhurst1999}; Coffey subsequently proved and extended the tetrahedron identities \cite{Coffey2008}.  The equal-mass $\operatorname{Li}_3(1/4)$ value below realizes the next weight in that three-dimensional ladder.

The proof has two stages.  First, the angular variables collapse to an odd trilogarithm of one radial cross-ratio.  Second, the 24 radial orderings are summed before the final integration.  This order of operations cancels the quadratic letters that occur sector by sector and leaves a real rational-logarithmic kernel.  A universal classical primitive then gives the eighteen-block expression.

\section{Angular reduction and a real one-dimensional kernel}
Let $r_i=|\mathbf r_i|$, $r_{ij}=|\mathbf r_i-\mathbf r_j|$, and set
% TRACE: Inc-03 angular collapse; Inc-17 real kernel.
\begin{equation}\label{eq:def-I}
 I(\alpha_1,\alpha_2,\alpha_3,\alpha_4)=
 \int_{\R^{12}}\frac{e^{-\sum_{i=1}^4\alpha_i r_i}}
 {r_1r_2r_3r_4\,r_{12}r_{23}r_{34}r_{41}}
 \prod_{i=1}^4\dd^3r_i .
\end{equation}
The integral is finite for $\alpha_i>0$ and is homogeneous of degree $-4$.

Put $M_{ij}=\max(r_i,r_j)$ and
% TRACE: Inc-03/17.
\begin{equation}\label{eq:q-radial}
 q(r)=\frac{r_1r_2r_3r_4}{M_{12}M_{23}M_{34}M_{41}}\in(0,1].
\end{equation}
Successive Legendre contractions around the cycle give
% TRACE: Inc-03, starting radial formula.
\begin{equation}\label{eq:angular-collapse}
 I(\alpha)=\frac{(4\pi)^4}{2}\int_{(0,\infty)^4}e^{-\sum_i\alpha_i r_i}
 \bigl\{\operatorname{Li}_3(q)-\operatorname{Li}_3(-q)\bigr\}
 \prod_{i=1}^4\dd r_i .
\end{equation}
This identity also proves positivity and permits all subsequent exchanges by Tonelli's theorem.

\subsection{Ordered radial sectors}
For a strict ordering
$r_{\sigma(1)}<r_{\sigma(2)}<r_{\sigma(3)}<r_{\sigma(4)}$, put
\[
 (r_{\sigma(1)},r_{\sigma(2)},r_{\sigma(3)},r_{\sigma(4)})
 =(uvwR,vwR,wR,R),\qquad 0<u,v,w<1.
\]
The Jacobian is $R^3vw^2$, and, with the ordered exponents
$(a,b,c,d)=(\alpha_{\sigma(1)},\alpha_{\sigma(2)},
\alpha_{\sigma(3)},\alpha_{\sigma(4)})$,
\[
 D_\sigma=d+w(c+bv+auv),\qquad
 \int_0^\infty R^3e^{-D_\sigma R}\,\dd R=\frac6{D_\sigma^4}.
\]
Let $m(\sigma)=1$ when the vertices carrying the two largest radii are adjacent
on $C_4$, and let $m(\sigma)=2$ when they are opposite.  There are sixteen
sectors of the first kind and eight of the second.  Thus
% TRACE: Inc-03 sector normal form and 16/8 pattern count.
\begin{equation}\label{eq:sector-normal}
 I(\alpha)=3(4\pi)^4\sum_{\sigma\in S_4}\int_{[0,1]^3}
 \frac{vw^2}{D_\sigma^4}
 \left\{\operatorname{Li}_3(uv^{m(\sigma)}w)
 -\operatorname{Li}_3(-uv^{m(\sigma)}w)\right\}\dd u\,\dd v\,\dd w.
\end{equation}
The two sector types are separately linearly reducible.  A sectorwise reduction,
however, introduces unnecessary quadratic letters.  The shorter general answer
is obtained by summing the physical opposite-pair sectors before the final
integration.

For the complete radial sum write $S=\sum_i\alpha_i$, $P=\prod_i\alpha_i$, and
$r=\alpha_1\alpha_3/(\alpha_2\alpha_4)$.  Define
% TRACE: Inc-17 rational-density theorem.
\begin{align}\label{eq:FA-Xi}
 F_A(x)&=\frac{[\alpha_1+(S-\alpha_1)x][\alpha_3+(S-\alpha_3)x]}
 {x[S-\alpha_2+\alpha_2x][S-\alpha_4+\alpha_4x]},\notag\\
 F_B(x)&=\frac{[\alpha_2+(S-\alpha_2)x][\alpha_4+(S-\alpha_4)x]}
 {x[S-\alpha_1+\alpha_1x][S-\alpha_3+\alpha_3x]},\notag\\
 \Xi(x)&=\frac{\prod_i(S-\alpha_i+\alpha_i x)}
 {\prod_{ij\in C_4}[\alpha_i+\alpha_j+(S-\alpha_i-\alpha_j)x]} .
\end{align}
The apparent poles in
% TRACE: Inc-17, eqs. for g and divided logarithms.
\begin{equation}\label{eq:g-kernel}
 g(x)=\frac{\log F_A(x)}{1-x/r}+\frac{\log F_B(x)}{1-rx}+\log\Xi(x)
\end{equation}
are removable.  For example, if $\Delta=\alpha_1\alpha_3-\alpha_2\alpha_4x$, then
$F_A(x)=1+\Delta h_A(x)$ with
$h_A=(x-1)^2/[x(S-\alpha_2+\alpha_2x)(S-\alpha_4+\alpha_4x)]$; hence the first quotient is a real analytic divided logarithm.  In terms of the logarithmic
radial phase $t$, the complete density is elementary:
% TRACE: Inc-17 elementary Borel density.
\begin{equation}\label{eq:elementary-density}
 b_\alpha(t)=\frac1{6P}\frac{\dd^2}{\dd t^2}g(e^t),\qquad
 \frac{I(\alpha)}{(4\pi)^4}=6\int_0^\infty b_\alpha(t)
 \frac{\operatorname{Li}_3(e^{-t})-\operatorname{Li}_3(-e^{-t})}{2}\,\dd t.
\end{equation}
The identities $g(1)=0$, $g'(1)=-1$, and
$g(x)-g(1/x)=-2\log x$ fix the endpoint constants and the reversal symmetry.

\begin{proposition}\label{prop:real-kernel}
For every $\alpha_i>0$,
% TRACE: Inc-17 positive real kernel.
\begin{equation}\label{eq:real-kernel}
 I(\alpha)=\frac{(4\pi)^4}{P}\int_0^1
 [g(x)+\log x]\,\frac{\log(1+x)-\log(1-x)}{2x}\,\dd x .
\end{equation}
The integrand is real and positive.  It grows only logarithmically at $0$ and is
$O((1-x)^4|\log(1-x)|)$ at $1$.
\end{proposition}

\begin{proof}
Equation~\eqref{eq:sector-normal} becomes a Laplace transform in the three logarithmic gaps.  Sixteen orderings have phase $g_1+g_2+g_3$ and eight have phase $g_1+2g_2+g_3$.  In a type-two cone the individual rational primitive contains a quadratic denominator.  Interchanging the two masses in the smaller opposite pair cancels its nonlogarithmic primitive; interchanging the larger pair and combining the two halves at the internal square root cancels the remaining square-root terms.

Writing $q=e^t$ and $\Delta=\alpha_1\alpha_3-\alpha_2\alpha_4q$, the surviving logarithmic coefficient is fixed by the elementary identity
\[
 \frac{\dd}{\dd q}\frac{q}{\Delta^2}
 =\frac{\alpha_1\alpha_3+\alpha_2\alpha_4q}{\Delta^3}.
\]
The first opposite-pair assembly is therefore rational in $q$ plus a rational multiple of $\log F_A(q)$; the other gives the same expression with $A$ and $B$ exchanged.  In the remaining rational part all simple-pole residues cancel, and its second primitive is $\log\Xi(q)$.  Adding the type-one sectors gives exactly \eqref{eq:elementary-density}.

The two endpoint constants are $g(1)=0$ and $g'(1)=-1$.  Two integrations by parts against the odd trilogarithm then give \eqref{eq:real-kernel}.  The divided-logarithm factorization above supplies continuation through $x=r$ and $x=1/r$, while the endpoint estimates follow from the displayed affine factors.  Finally, put $h(t)=g(e^t)+t$.  The sector representation gives $h(-t)=h(t)$, $h(0)=h'(0)=0$, and $h''(t)>0$ for $t\ne0$; hence the integrand in \eqref{eq:real-kernel} is positive.
\end{proof}

\section{The classical-polylogarithmic formula}
This section turns \eqref{eq:real-kernel} into a finite expression.  Every object in the theorem is defined by algebraic operations and endpoint substitutions; no unevaluated integral is retained.

Normalize
% TRACE: Inc-18 eq. normalization.
\begin{equation}\label{eq:norm}
 S=\sum_i\alpha_i,
 \qquad P=\prod_i\alpha_i,
 \qquad \mu_i=\frac{\alpha_i}{S},
 \qquad r=\frac{\mu_1\mu_3}{\mu_2\mu_4}.
\end{equation}
With indices modulo four, define the affine functions
% TRACE: Inc-18 affines.
\begin{equation}\label{eq:affines}
 L_i(x)=\mu_i+(1-\mu_i)x,
 \quad M_i(x)=1-\mu_i+\mu_i x,
 \quad E_i(x)=\mu_i+\mu_{i+1}+(1-\mu_i-\mu_{i+1})x,
 \quad X(x)=x.
\end{equation}
The functions $L_i,M_i,E_i$ are positive on $[0,1]$, while $X$ is positive on $(0,1]$ and has a simple zero at $0$.

For $s=2,3$, let $\operatorname{Li}_s^{\R}$ denote the real continuation on the real axis.  On $(-1,1)$ it is the ordinary power-series value, with the usual values at $0$ and $\pm1$.  For $y>1$ it is defined by
% TRACE: Inc-18 real inversion identities.
\begin{equation}\label{eq:li-pos-2}
 \operatorname{Li}_2^{\R}(y)=\frac{\pi^2}{3}-\frac12\log^2y-\operatorname{Li}_2(1/y),
\end{equation}
% TRACE: Inc-18 real inversion identities.
\begin{equation}\label{eq:li-pos-3}
 \operatorname{Li}_3^{\R}(y)=\frac{\pi^2}{3}\log y-\frac16\log^3y+\operatorname{Li}_3(1/y),
\end{equation}
and for $y<-1$ by
% TRACE: Inc-18 real inversion identities.
\begin{equation}\label{eq:li-neg-2}
 \operatorname{Li}_2^{\R}(y)=-\frac{\pi^2}{6}-\frac12\log^2(-y)-\operatorname{Li}_2(1/y),
\end{equation}
% TRACE: Inc-18 real inversion identities.
\begin{equation}\label{eq:li-neg-3}
 \operatorname{Li}_3^{\R}(y)=-\frac{\pi^2}{6}\log(-y)-\frac16\log^3(-y)+\operatorname{Li}_3(1/y).
\end{equation}
These formulas define one real branch for every real argument.  They also move every final nonconstant polylogarithm argument into $(-1,1)$.

For $z>0$ and real $d\ne0$, set
% TRACE: Inc-18 universal primitive.
\begin{equation}\label{eq:phi1}
 \Phi_1(d,z)=\log z\log|1-z/d|+\operatorname{Li}_2^{\R}(z/d),
 \qquad \Phi_1(0,z)=\frac12\log^2z,
\end{equation}
% TRACE: Inc-18 universal primitive.
\begin{equation}\label{eq:phi2}
 \Phi_2(d,z)=\log^2z\log|1-z/d|+2\log z\operatorname{Li}_2^{\R}(z/d)
 -2\operatorname{Li}_3^{\R}(z/d),
 \qquad \Phi_2(0,z)=\frac13\log^3z.
\end{equation}
Direct differentiation gives
% TRACE: Inc-18 primitive derivative.
\begin{equation}\label{eq:phi-derivative}
 \frac{\partial}{\partial z}\Phi_j(d,z)=\frac{\log^jz}{z-d},\qquad j=1,2.
\end{equation}
Let $L=A+Bx$ and $M=C+Dx$ be affine functions that are positive on the open integration interval.  First suppose that $B,D,BC-AD$, and $M(c)$ are nonzero.  Define
% TRACE: Inc-18 Möbius-polarization primitive.
\begin{align}\label{eq:J-generic}
 \mathcal J_c(L,M;x)=\frac12\bigg\{&\Phi_2(L(c),L(x))+\Phi_2(M(c),M(x))\notag\\
 &-\Phi_2\!\left(\frac{L(c)}{M(c)},\frac{L(x)}{M(x)}\right)
 +\Phi_2\!\left(\frac BD,\frac{L(x)}{M(x)}\right)\bigg\}.
\end{align}
It is a primitive of $\log L(x)\log M(x)/(x-c)$.  If $M(c)=0$, the third term in \eqref{eq:J-generic} is absent; if $D=0$, the fourth is absent.  These are the corresponding finite limits of the differentiated identity, not numerical limiting prescriptions.  When $M=\kappa L$ with $\kappa>0$, use instead
% TRACE: Inc-18 proportional-affine branch.
\begin{equation}\label{eq:J-prop}
 \mathcal J_c(L,\kappa L;x)=\Phi_2(L(c),L(x))+\log\kappa\,\Phi_1(L(c),L(x)).
\end{equation}
If one affine function is constant, the corresponding limit is
$\log A\,\Phi_1(M(c),M(x))$; for two constants it is
$\log A\log C\log|x-c|$.

Endpoint collisions are finite substitutions.  At an endpoint write $x=\varepsilon$ or $x=1-\varepsilon$ and expand in $\tau=\log\varepsilon$.  The map $\mathcal E_e$, $e\in\{0,1\}$, extracts the coefficient of $\tau^0$ after all powers of $\tau$ have been retained.  Divergent powers cancel only after the complete eighteen-block expression is assembled.  The only nontrivial local row needed below is
% TRACE: Inc-18 endpoint table.
\begin{equation}\label{eq:end-collision}
 \begin{aligned}
 z=d+q\varepsilon+O(\varepsilon^2),\ d>0:\qquad
 \mathcal E_e\Phi_1&=\log d\log|q/d|+\frac{\pi^2}{6},\\
 \mathcal E_e\Phi_2&=\log^2d\log|q/d|+\frac{\pi^2}{3}\log d-2\zeta(3).
 \end{aligned}
\end{equation}
For completeness, the endpoint maps are given by the following finite table.
Here $z\sim C\varepsilon^k$ means $C>0$, and $D_0=\log|d|$.
\begin{center}
\small
\begin{tabular}{@{}p{0.28\linewidth}p{0.30\linewidth}p{0.34\linewidth}@{}}
\toprule
case & $\mathcal E_e\Phi_1(d,z)$ & $\mathcal E_e\Phi_2(d,z)$\\\midrule
$d=0$, $z\sim C\varepsilon^k$ & $\tfrac12\log^2 C$ & $\tfrac13\log^3 C$\\
$d\ne0$, $z\to0$ & $0$ & $0$\\
$d\ne0$, $z\sim C/\varepsilon$ & $\tfrac12\log^2C-\tfrac12D_0^2+C_2(d)$ & $\tfrac13\log^3C-\tfrac13D_0^3+2C_2(d)D_0$\\
$z\to z_0>0$, $z_0\ne d$ & $\Phi_1(d,z_0)$ & $\Phi_2(d,z_0)$\\
$z=d+q\varepsilon+O(\varepsilon^2)$, $d>0$ & $\log d\log|q/d|+\pi^2/6$ & $\log^2d\log|q/d|+(\pi^2/3)\log d-2\zeta(3)$\\
\bottomrule
\end{tabular}
\end{center}
where $C_2(d)=\pi^2/3$ for $d>0$ and $C_2(d)=-\pi^2/6$ for $d<0$.
The complete endpoint expansions are polynomials in $\log\varepsilon$ plus terms tending to zero.  Their six divergent coefficients cancel exactly after the eighteen blocks are assembled.  The table therefore defines $\mathcal E_0$ and $\mathcal E_1$ by finite algebra, without numerical limiting procedures.

Define the finite classical block
% TRACE: Inc-18 definition of T block.
\begin{equation}\label{eq:T-block}
 T_c(L)=(\mathcal E_1-\mathcal E_0)
 \bigl[\mathcal J_c(1+x,L;x)-\mathcal J_c(1-x,L;x)\bigr].
\end{equation}
Finally put
% TRACE: Inc-18 18-block assembly.
\begin{align}\label{eq:S-assembly}
 \mathcal S={}&\sum_{i=1}^4T_0(L_i)-\sum_{i=1}^4T_0(E_i)\notag\\
 &-T_r(L_1)-T_r(L_3)+T_r(M_2)+T_r(M_4)+T_r(X)\notag\\
 &-T_{1/r}(L_2)-T_{1/r}(L_4)+T_{1/r}(M_1)+T_{1/r}(M_3)+T_{1/r}(X).
\end{align}
There are eighteen signed block evaluations before coincident terms are collected.

\begin{theorem}[Classical evaluation]\label{thm:main}
For every $\alpha_1,\ldots,\alpha_4>0$,
% TRACE: Inc-18 main theorem; deliberately equation 3.15.
\begin{equation}\label{eq:main-classical}
 \boxed{
 I(\alpha_1,\alpha_2,\alpha_3,\alpha_4)
 =\frac{(4\pi)^4}{P}\left\{\frac78\zeta(3)+\frac12\mathcal S\right\}.}
\end{equation}
This is a finite real expression in $\log$, $\operatorname{Li}_2$, and
$\operatorname{Li}_3$ of rational functions of the masses.  It has weight at most three and depth one.  No minimality statement is made.
\end{theorem}

\begin{proof}
On $0<x<1$, the three logarithms in \eqref{eq:g-kernel} decompose as
\begin{align*}
 \log F_A&=\log L_1+\log L_3-\log M_2-\log M_4-\log x,\\
 \log F_B&=\log L_2+\log L_4-\log M_1-\log M_3-\log x,\\
 \log\Xi&=\sum_{i=1}^4\log M_i-\sum_{i=1}^4\log E_i.
\end{align*}
The polarization identity
$\log L\log M=\tfrac12\{\log^2L+\log^2M-\log^2(L/M)\}$, followed by $z=L/M$, gives \eqref{eq:J-generic}; indeed
$\dd x/(x-c)=\dd z/(z-L(c)/M(c))-\dd z/(z-B/D)$.  Formula \eqref{eq:phi-derivative} proves the primitive, including the finite degenerate branches described above.

Using
\[
 \frac{r}{x(r-x)}=\frac1x-\frac1{x-r},
 \qquad
 \frac{1}{x(1-rx)}=\frac1x-\frac1{x-1/r},
\]
and applying the endpoint maps term by term, the displayed logarithmic decompositions give
\[
 \mathcal S=\int_0^1
 \frac{\log(1+x)-\log(1-x)}{x}\,[g(x)+2\log x] \,\dd x.
\]
The six potentially divergent endpoint-log coefficients cancel identically in this assembled expression.  Finally,
\[
 \int_0^1\log x\,[\log(1+x)-\log(1-x)]\frac{\dd x}{x}
 =-\frac74\zeta(3).
\]
Substitution into the preceding identity converts \eqref{eq:real-kernel} into \eqref{eq:main-classical}.  Equations \eqref{eq:li-pos-2}--\eqref{eq:li-neg-3} then move every remaining argument to $(-1,1)$ without changing the real value.
\end{proof}

\section{Equal exponents}
The general formula simplifies substantially on the diagonal.  Direct integration of the equal-mass elementary density gives
% TRACE: Inc-16 two-trilogarithm formula.
\begin{align}\label{eq:equal-two-trilog}
 \frac{I(1,1,1,1)}{(4\pi)^4}={}&2\operatorname{Li}_3\!\left(\frac23\right)
 -\operatorname{Li}_3\!\left(\frac34\right)-\frac78\zeta(3)-\frac{\pi^2}{12}\log2\notag\\
 &+\left(\frac{\pi^2}{3}-\log^22\right)\log\frac32
 -\frac13\log^3\frac32.
\end{align}

\begin{proposition}\label{prop:li3-compression}
Identity
% TRACE: Inc-16 plus exact Lewin/Spence-Kummer compression proved in this revision.
\begin{align}\label{eq:li3-identity}
&2\operatorname{Li}_3\!\left(\frac23\right)-\operatorname{Li}_3\!\left(\frac34\right)
 -\frac78\zeta(3)-\frac{\pi^2}{12}\log2
 +\left(\frac{\pi^2}{3}-\log^22\right)\log\frac32
 -\frac13\log^3\frac32\notag\\
&\hspace{20mm}=\frac12\operatorname{Li}_3\!\left(\frac14\right)
 -\frac7{24}\zeta(3)+\frac{\pi^2}{12}\log2-\frac23\log^32
\end{align}
holds.
\end{proposition}

\begin{proof}
Apply the three-term relation
% TRACE: Lewin 1981, three-term trilogarithm relation.
\begin{equation}\label{eq:li3-three-term}
 \operatorname{Li}_3(x)+\operatorname{Li}_3(1-x)+\operatorname{Li}_3(1-1/x)
 =\zeta(3)+\frac16\log^3x+\frac{\pi^2}{6}\log x
 -\frac12\log^2x\log(1-x)
\end{equation}
to $x=3/4$ and $x=2/3$, and use
$\operatorname{Li}_3(z)+\operatorname{Li}_3(-z)=\tfrac14\operatorname{Li}_3(z^2)$
at $z=1/2$ and $z=1/3$.  If the right side of \eqref{eq:li3-three-term} is denoted by $e(x)$, then
% TRACE: symbolic intermediate combination from Proposition 4.1 proof.
\begin{align}\label{eq:li3-combination}
6\operatorname{Li}_3\!\left(\frac14\right)-24\operatorname{Li}_3\!\left(\frac23\right)
+12\operatorname{Li}_3\!\left(\frac34\right)
={}&12e\!\left(\frac34\right)-24e\!\left(\frac23\right)\notag\\
&-24\operatorname{Li}_3\!\left(\frac12\right)
+36\operatorname{Li}_3\!\left(\frac13\right)
-3\operatorname{Li}_3\!\left(\frac19\right).
\end{align}
Insert
% TRACE: standard Li3(1/2) and base-3 Spence-Kummer identity.
\begin{align}\label{eq:li3-base-values}
 \operatorname{Li}_3\!\left(\frac12\right)&=\frac78\zeta(3)-\frac{\pi^2}{12}\log2+\frac16\log^32,\notag\\
 12\operatorname{Li}_3\!\left(\frac13\right)-\operatorname{Li}_3\!\left(\frac19\right)
 &=\frac{26}{3}\zeta(3)-\frac23\pi^2\log3+\frac23\log^33.
\end{align}
The second line is the tabulated base-three identity \cite[Eq.~(E.14)]{DelDucaDuhrSmirnov}; it follows by combining duplication with
\[
 \operatorname{Li}_3\!\left(-\frac13\right)=2\operatorname{Li}_3\!\left(\frac13\right)
 -\frac{13}{6}\zeta(3)+\frac{\pi^2}{6}\log3-\frac16\log^33,
\]
the specialization $(x,y)=(-1,1/3)$ of the Spence--Kummer nine-term relation \cite[Eq.~(6.107)]{Lewin}; see also \cite[Sec.~5]{Broadhurst2010}.  Writing $\ell_2=\log2$ and $\ell_3=\log3$, the right side of
\eqref{eq:li3-combination} becomes, after substitution,
\[
24\ell_2^3-24\ell_2^2\ell_3+12\ell_2\ell_3^2-4\ell_3^3
-6\pi^2\ell_2+4\pi^2\ell_3-7\zeta(3).
\]
Solving \eqref{eq:li3-combination} for
$2\operatorname{Li}_3(2/3)-\operatorname{Li}_3(3/4)$ and then using
$\log(3/2)=\ell_3-\ell_2$ gives exactly the right side of
\eqref{eq:li3-identity}.  This is an algebraic collection of the displayed
functional equations, not a numerical verification.
\end{proof}

\begin{corollary}\label{cor:equal}
For every $\lambda>0$,
% TRACE: Inc-16/18 equal-mass formula and degree -4 homogeneity.
\begin{equation}\label{eq:equal-final}
 \boxed{
 I(\lambda,\lambda,\lambda,\lambda)=\frac{(4\pi)^4}{\lambda^4}
 \left\{\frac12\operatorname{Li}_3\!\left(\frac14\right)-\frac7{24}\zeta(3)
 +\frac{\pi^2\log2}{12}-\frac23\log^32\right\}.}
\end{equation}
\end{corollary}
The relation to Broadhurst's three-dimensional tetrahedron is visible at the level of weight: the Clausen/dilogarithmic tetrahedron is followed here, one loop up, by a trilogarithmic constant with the same low-prime ramification.

\section{Verification and structural checks}
The analytic proof is accompanied by three independent numerical routes.  Their common input is only the angular reduction \eqref{eq:angular-collapse}; their subsequent numerical algorithms are different.

\subsection{Three numerical routes}
The original-sector route integrates the literal summands in
\eqref{eq:sector-normal}.  A dyadic decomposition isolates the corner
$u=v=w=1$.  On every retained box, complex-neighborhood bounds for the
three coordinate variables give a certified tensor Gauss remainder; the final
corner is bounded directly by positivity.  The complete calculation used 334
boxes and more than five million distinct tensor-grid points.

The positive-series route expands the odd trilogarithm and evaluates, for
$n=2\ell+1$,
\[
 R_m(n)=\int_{[0,1]^3}\frac{u^n v^{mn+1}w^{n+2}}
 {[d+w(c+bv+auv)]^4}\,\dd u\,\dd v\,\dd w.
\]
It then sums
\[
 I(\alpha)=6(4\pi)^4\sum_{\sigma\in S_4}\sum_{\ell\ge0}
 \frac{R_{m(\sigma)}(2\ell+1)}{(2\ell+1)^3}.
\]
Every term is positive.  A beta expansion evaluates retained modes in exact
rational arithmetic, while an explicit finite inverse-mode expansion and
Euler--Maclaurin bounds enclose the complete tail.

The classical route compiles Theorem~\ref{thm:main} at an exact rational mass
tuple, removes all endpoint symbols, and evaluates the resulting real
logarithms and polylogarithms in exact dyadic interval arithmetic.  Every
infinite series carries a signed remainder; no numerical quadrature or GPL
evaluator enters this route.

\begin{table}[ht]
\centering\small
\begin{tabular}{@{}llr@{}}
\toprule
$(\alpha_1,\alpha_2,\alpha_3,\alpha_4)$ & common certified value & common places\\\midrule
$(1,1,1,1)$ & $3159.60899251365163791156042906107837$ & 32\\
$(1,2,3,5)$ & $78.31045350989936985851885309501456$ & 32\\
$(1,1,2,2)$ & $716.43823041968417808646698360606176$ & 32\\
$(1.1,1.85,2.37,2.91)$ & $200.75952442512883657208388007939669$ & 32\\
\bottomrule
\end{tabular}
\caption{Agreement of original-sector quadrature, the positive angular-mode series, and the classical formula.}
\label{tab:three-routes}
\end{table}
% TRACE: Inc-05,07--09,13,18 benchmark certificates; digit strings in Table 1.

The exact-dyadic evaluator supplies 100-place intervals at eight tuples.  The common rounded values are:
\begin{center}
\refstepcounter{table}\label{tab:hundred}
\scriptsize
\begin{tabular}{@{}lp{0.73\linewidth}@{}}
\toprule
masses & certified 100-place value\\\midrule
$(1,1,1,1)$ & {\ttfamily 3159.608992513651637911560429061078369625096312248831872\newline 8322738025231878407608611943955574925132846650851}\\
$(1,2,3,5)$ & {\ttfamily 78.31045350989936985851885309501456186993083754527936745\newline 94967561146202833166651934179999373512458506338}\\
$(1,1,2,2)$ & {\ttfamily 716.4382304196841780864669836060617632616949369595853072\newline 457888898113121420979852100909612952690407001200}\\
King tuple & {\ttfamily 200.7595244251288365720838800793966851018292482414653781\newline 531291922000696263638073637767621906002300833931}\\
$(1,2,6,3)$ & {\ttfamily 62.70307719275218832581947231547944247214996829556567001\newline 71626699078532490873784853026176303437132045940}\\
$(2,3,5,7)$ & {\ttfamily 12.40569428355165726553345288935811758944101442733473916\newline 10718067048825066600988476390504644825884383447}\\
$(1,2,3,6)$ & {\ttfamily 60.94082285700569606527527108607063062454645689314468954\newline 50520039757845211518365881010561149487097398303}\\
$(2,1,3,5)$ & {\ttfamily 80.03981495875086293025715948477197523691716004857586966\newline 80829768513099706722234321251439589677881046779}\\
\bottomrule
\end{tabular}
\par\smallskip
\small Table~\thetable: Exact-dyadic certificates.  The King tuple is $(11/10,37/20,237/100,291/100)$.
\end{center}
% TRACE: Inc-18 eight 100-place certificates; all digit strings in Table 2.

As an external calibration, the positive radial machinery reproduces King's Table~2 values at the King tuple to every printed digit:
\begin{table}[ht]
\centering\small
\begin{tabular}{@{}lp{0.53\linewidth}@{}}
\toprule
integral & reproduced value\\\midrule
path $12,23,34$ & $155.2788229172762030448705061992562042709$\\
star, center-spoke derivative ($S'_4$) & $73.89819601667390396396288267469508944573$\\
triangle $12,13,23$ with leaf $34$ ($L_4$) & $43.79915769725917539512080204618774228423$\\
\bottomrule
\end{tabular}
\caption{Independent reproduction of the three calibration entries in \cite[Table~2]{King2014}.}
\label{tab:king}
\end{table}
% TRACE: Inc-06 King calibration certificate; all digit strings in Table 3.

The cycle has the dihedral symmetry group $D_4$, not the full symmetric group.  As a certified non-$D_4$ control, the two interval calculations give, to 35 decimal places,
% TRACE: Inc-14 non-D4 negative control.
\begin{equation}\label{eq:negative-control}
 I(2,1,3,5)-I(1,2,3,5)
 \approx1.72936144885149307173830638975741337>0.
\end{equation}

\begin{theorem}[Rearrangement law]\label{thm:rearrangement}
For $a,b,c,d>0$,
% TRACE: Inc-14 universal rearrangement theorem.
\begin{equation}\label{eq:swap-law}
 \sgn\{I(b,a,c,d)-I(a,b,c,d)\}=\sgn\{(b-a)(d-c)\}.
\end{equation}
Equality holds exactly when $a=b$ or $c=d$.  Consequently, if $0<a\le b\le c\le d$, then
% TRACE: Inc-14 arrangement order and equality cases.
\begin{equation}\label{eq:arrangement-order}
 I(a,c,b,d)\le I(a,b,c,d)\le I(a,b,d,c),
\end{equation}
with strict inequalities when the four masses are strictly ordered.
\end{theorem}

\begin{proof}
On $0<x_1<x_2<x_3<x_4$, the two cycle arguments are $x_1/x_4$ and $x_1x_2/(x_3x_4)$; the first is larger, so their odd-trilogarithm difference is positive.  The difference between two mass assignments factors into sums of products of $2\times2$ determinants of the exponential matrix $e^{-\alpha_i x_j}$.  Each determinant has the sign opposite to its mass difference.  Integration against the positive radial measure gives \eqref{eq:swap-law}; strictness and the equality cases follow from the strict determinant signs off coincident masses.
\end{proof}

\subsection{Derivative consistency}
Differentiation under the integral is justified uniformly on compact subsets of
the positive mass orthant.  Direct radius-insertion calculations at
$(1,2,3,5)$ give the complete gradient
\begin{center}
\small
\begin{tabular}{@{}lc@{}}
\toprule
quantity & certified value (40 places)\\\midrule
$-\partial_{\alpha_1}I$ & $51.5938232136790932412965529908231246727304$\\
$-\partial_{\alpha_2}I$ & $37.7244391882592544581271087453801382787620$\\
$-\partial_{\alpha_3}I$ & $26.9616321819670585425731703126374880451995$\\
$-\partial_{\alpha_4}I$ & $21.0628431806997403297610261921124764227741$\\
\bottomrule
\end{tabular}
\end{center}
% TRACE: Inc-11/15 direct marked-radius gradient certificates.
The four intervals satisfy the homogeneous Euler identity
$\sum_i\alpha_i(-\partial_{\alpha_i}I)=4I$.  At both equal and asymmetric
masses, a separately bounded central difference of the finite classical formula
lies inside the direct marked-integral interval.  Thus the derivative test does
not rely on symmetry or on formal differentiation of a numerically fitted value.

\subsection{Mass limits}
For the small-mass limit define the elementary three-electron path integral
\[
 T(a,b,c)=\int_{\R^9}
 \frac{e^{-ar_1-br_2-cr_3}}{r_1r_2r_3\,r_{12}r_{23}}
 \prod_{i=1}^3\dd^3r_i.
\]
Its closed form is
% TRACE: Inc-12 path formula.
\begin{equation}\label{eq:T-small}
 T(a,b,c)=\frac{(4\pi)^3}{a^2c^2}
 \log\frac{(a+b)(b+c)}{b(a+b+c)}.
\end{equation}
Then, with the reference mass $b$ making the logarithm dimensionless,
% TRACE: Inc-12 logarithmic small-mass law.
\begin{equation}\label{eq:small-law}
 \lim_{\varepsilon\downarrow0}\frac{I(a,b,c,\varepsilon)}{\log(b/\varepsilon)}=4\pi T(a,b,c).
\end{equation}
The proof isolates the exact inverse-radius tail of the odd trilogarithm after the fourth radius exceeds the other two neighboring radii.  More precisely, with
$L=4\pi T(a,b,c)$ and $K=(4\pi)^4/(a^2bc^2)$ there is a finite constant $C_b$ such that
% TRACE: Inc-12 global one-sided small-mass remainder.
\begin{equation}\label{eq:small-quantitative}
 0\le I(a,b,c,\varepsilon)-L\log(b/\varepsilon)-C_b\le K\varepsilon
 \qquad(\varepsilon>0).
\end{equation}
There is also a second-order logarithmic remainder after the explicit linear
term is removed.  These bounds justify the limit independently of numerical
extrapolation.

For the opposite limit, the coefficient is the finite three-electron integral
\[
 U(a,b,c)=\int_{\R^9}
 \frac{e^{-ar_1-br_2-cr_3}}{r_1^2r_2r_3^2\,r_{12}r_{23}}
 \prod_{i=1}^3\dd^3r_i.
\]
Put $X=a/(a+b)$, $Y=c/(b+c)$, $S_3=a+b+c$, and
% TRACE: Inc-11 dilogarithmic U coefficient.
\begin{equation}\label{eq:V-large}
 V=\frac{\pi^2}{6}+\operatorname{Li}_2(XY)-\operatorname{Li}_2(X)-\operatorname{Li}_2(Y)+\log X\log Y.
\end{equation}
Define
% TRACE: Inc-11 closed U formula.
\begin{equation}\label{eq:U-large}
 \frac{U(a,b,c)}{(4\pi)^3}=\frac{V}{b^2}
 +\frac1{ab}\log\frac{b+c}{S_3}+\frac1{bc}\log\frac{a+b}{S_3}
 +\frac1{ac}\log\frac{(a+b)(b+c)}{bS_3}.
\end{equation}
Then
% TRACE: Inc-11 large-mass law and no logarithm.
\begin{equation}\label{eq:large-law}
 I(a,b,c,A)=\frac{4\pi}{A^2}U(a,b,c)+o(A^{-2}),\qquad A\to\infty.
\end{equation}
Unlike the complete six-edge rim, the cycle has no logarithm at this order.  If
$m_*=\min(a,b,c)$, the finite-$A$ error is bounded explicitly by
% TRACE: Inc-11 quantitative large-mass bounds.
\begin{align}\label{eq:large-bounds}
-\frac{(4\pi)^3}{A}\left(\frac1a+\frac1c\right)
\left[4+2\log\left(1+\frac Ab\right)\right]
&\le \frac{A^2I(a,b,c,A)}{4\pi}-U(a,b,c)\notag\\
&\le \frac{3(4\pi)^3}{4Am_*}.
\end{align}
Thus \eqref{eq:large-law} follows with a fully explicit error term.

\section{Scope of applicability}
The formula evaluates the Slater integral with all four nuclear spokes and the four interelectronic edges of a cycle.  Differentiation with respect to the $\alpha_i$ supplies nonnegative powers of the one-electron radii.  Standard Hy-CI, in the conventional sense reviewed by Harris, appends at most one $r_{ij}$ to each basis term and therefore never generates a cyclic rim.  The direct consumers are fully correlated Hylleraas expansions, double-linked Hy-CI, and extended Hy-CI (E-Hy-CI).  No assertion is made here about the efficiency or convergence rate of any particular atomic calculation, or about integrals containing either diagonal.

\section{Outlook: the first chord}\label{sec:outlook}
Adding the chord $r_{13}^{-1}$ gives the rim $K_4-\{24\}$.  At fixed radii its differentiated angular kernel is the Lauricella period $F_D^{(3)}(1;\tfrac12,\tfrac12,\tfrac12;\tfrac32;X,Y,Z)$; a one-parameter specialization has a rank-two monodromy representation whose Zariski closure contains $\mathrm{SL}_2$, and the fixed-radius kernel is therefore not a multiple polylogarithm of the declared regularized class.  Whether the elliptic quotient survives the physical mass-only radial assembly remains open.  Details of this chorded analysis are to appear separately.
% TRACE: R3-D1/R3-D2 fixed-radius elliptic and Lauricella results; mass-level status explicitly open.

\appendix
\section{Scope of the numerical certificates}
The numerical records accompanying the derivation check the displayed formula by three routes, but they do not replace the analytic reductions in Sections~2--4.  The original-sector quadrature applies directly to the integrand in \eqref{eq:sector-normal}.  If
$C_j=[1-2^{-j},1]^3$, its only omitted box is $C_{46}$.  With
$\delta=2^{-46}$ and $\rho=1-\delta$, positivity gives, in an ordered sector,
% TRACE: Inc-13 omitted-corner bound.
\[
 0\le \int_{C_{46}} f_m\le
 \frac{3\delta^3}{[d+c\rho+b\rho^2+a\rho^3]^4}.
\]
On every retained box $B$, explicit complex-neighborhood bounds
$M_u,M_v,M_w$ give the tensor Gauss estimate
% TRACE: Inc-13 certified Gauss remainder.
\[
 \left|\int_B f_m-Q_{n,B}f_m\right|
 \le \frac{28}{5}|B|(M_u+M_v+M_w)\left(\frac4{49}\right)^n.
\]
All node, weight, special-function, and rounding errors are enclosed separately.
The resulting four complete intervals have widths between
$2.8\times10^{-36}$ and $1.3\times10^{-34}$.

The positive angular-mode computation encloses each retained beta expansion in
exact rational arithmetic and bounds all omitted modes analytically.  The
classical-formula computation is independent of that recurrence.  For
$0<x<1$ it uses, among other signed estimates,
% TRACE: Inc-18 elementary polylogarithm tail.
\[
 0\le \operatorname{Li}_s(x)-\sum_{k=1}^{N}\frac{x^k}{k^s}
 \le \frac{x^{N+1}}{(N+1)^s(1-x)},
\]
with the alternating next-term bound for negative arguments.  Logarithms,
$\pi$, and $\zeta(3)$ are enclosed by rational series with outward rounding.
The eight intervals in Table~\ref{tab:hundred} have widths below
$1.22\times10^{-137}$.  The scripts verify finite algebra, interval
propagation, sector enumeration, and the printed decimal enclosures; they do
not establish the literature history, the analytic continuation arguments, or
any claim for a graph containing a chord.

\begingroup\raggedright
\small
\setlength{\itemsep}{2pt}

\endgroup
\end{document}